\documentclass[a4paper]{article}
\usepackage{jheppub}
\usepackage[utf8]{inputenc}
\usepackage{amsmath,amsthm,amssymb,mathtools,tikz,booktabs}

\newcommand{\mc}{\mathcal} \newcommand{\mf}{\mathfrak} \newcommand{\mb}{\mathbb} \newcommand{\on}{\operatorname} \renewcommand{\r}{\mathbf} \newcommand{\la}{\langle} \newcommand{\ra}{\rangle}
\newcommand{\w}[1]{\wedge^{\!#1}\,} \newcommand{\slot}{\;\cdot\;} 
\title{Y-algebroids and $E_{7(7)} \times \mb{R}^+$-generalised geometry}

\usepackage{thmtools}
\declaretheorem[name=Theorem, parent=section]{thm}
\declaretheorem[name=Definition, sibling=thm]{defn}
\declaretheorem[name=Proposition, sibling=thm]{prop}

\usepackage{tikz-cd}

\usetikzlibrary{matrix,arrows,decorations.pathmorphing,decorations.markings,calc,patterns}
\tikzset{->-/.style={decoration={
  markings,
  mark=at position #1 with {\arrow{stealth'}}},postaction={decorate}}}
\tikzset{proj_empty/.style={circle,fill=white,inner sep=0pt}}
\tikzset{proj0/.style={circle,fill=white,draw=black,inner sep=0pt}}
\tikzset{proj1/.style={circle,fill=black,draw=black,inner sep=0pt}}
\tikzset{proj2/.style={circle,fill=blue,draw=blue,inner sep=0pt}}
\tikzset{proj3/.style={circle,fill=gray,draw=gray,inner sep=0pt}}

\makeatletter \gdef\@fpheader{} \makeatother
\subheader{\hfill\textrm{HU-EP-23/42}\\
\hspace*{\fill}\textrm{Imperial/TP/23/DW/1}}
\abstract{We define the notion of Y-algebroids, generalising the Lie, Courant, and exceptional algebroids that have been used to capture the local symmetry structure of type II string theory and M-theory compactifications to $D \geq 5$ dimensions. Instead of an invariant inner product, or its generalisation arising in exceptional algebroids, Y-algebroids are built around a specific type of tensor, denoted $Y$, that provides exactly the necessary properties to also describe compactifications to $D=4$ dimensions. We classify ``M-exact'' $E_7$-algebroids and show that this precisely matches the form of the generalised tangent space of $E_{7(7)} \times \mb{R}^+$-generalised geometry, with possible twists due to 1-, 4- and 7-form fluxes, corresponding physically to the derivative of the warp factor and the M-theory fluxes. We translate the notion of generalised Leibniz parallelisable spaces, relevant to consistent truncations, into this language, where they are mapped to so-called exceptional Manin pairs. We also show how to understand Poisson--Lie U-duality and exceptional complex structures using Y-algebroids.}
\author[a]{Ond\v rej Hulík,}
\author[b]{Emanuel Malek,}
\author[c]{Fridrich Valach,}
\author[c]{Daniel Waldram}
\affiliation[a]{Theoretische Natuurkunde, Vrije Universiteit Brussel\\ Pleinlaan 2, B-1050 Brussels, Belgium}
\affiliation[b]{Institut f\"{u}r Physik, Humboldt-Universit\"{a}t zu Berlin,\\
	IRIS Geb\"{a}ude, Zum Gro{\ss}en Windkanal 2, 12489 Berlin, Germany}
\affiliation[c]{Department of Physics, Imperial College London\\Prince Consort Road, London, SW7 2AZ, UK}
\emailAdd{ondra.hulik@gmail.com}
\emailAdd{emanuel.malek@physik.hu-berlin.de}
\emailAdd{fridrich.valach@gmail.com}
\emailAdd{d.waldram@imperial.ac.uk}

\begin{document}\maketitle

\section{Introduction}
Over the past few decades it has been understood that various aspects of string and M-theory
admit a convenient and elegant description in the language of generalised geometry. Notable examples include the symmetry algebra and dynamics of the massless NS-NS string sector in terms of ``$O(n,n)$-generalised geometry''~\cite{Hitchin,Gualtieri,CSCW0} (see also~\cite{Siegel,HHZ,JLP}) and of (the dimensional reductions of) the massless type II and M-theoretic sector in terms of ``$E_{n(n)}\times\mb R^+$-generalised geometry''~\cite{Hull,PW,CSCW} (otherwise known as ``exceptional generalised geometry''). The former constitutes one part of a larger world of Courant algebroids~\cite{LWX}, and this fact has been used to understand and extend the phenomenon of Poisson--Lie T-duality~\cite{KS,let,Scs}.
  
In order to develop an analogous extended version in the M-theoretic case, a class of geometric structures called exceptional algebroids was introduced and examined in~\cite{BHVW,BHVW2,HV}. This provided  a systematic geometrical construction and direct algebraic characterisation of all maximally supersymmetric consistent truncations to five dimensions and above, reformulating and slightly refining earlier results of Inverso~\cite{I}. It also gives new insights into the nature of Poisson--Lie U-duality~\cite{Sakatani,MT,MST}. Furthermore, it was shown that both Courant and exceptional algebroids, as well as the more well-known Lie algebroids, form specific subclasses of a more general class of structures, called $G$-algebroids. The letter $G$ here represents a choice of particular group data --- for instance, in the Lie, Courant, and exceptional case this data corresponds to the groups $GL(n,\mb R)$, $O(p,q)$, and $E_{n(n)}\times\mb R^+$, respectively. However, somewhat unsatisfactorily, exceptional algebroids (or more generally $G$-algebroids) only provide a description of the exceptional generalised geometry in the cases $n<7$.

The reason for this is as follows. The central ingredient in the definition of a $G$-algebroid is a ``bracket'', which encodes the algebra of symmetries of the related theory. Although by construction the infinitesimal symmetries themselves form a Lie algebra, it is more convenient to parametrise them in terms of sections of some vector bundle. The price to pay for this transition is that the corresponding bracket ceases (in general) to be antisymmetric, and only satisfies the Leibniz identity, a generalisation of the Jacobi identity to non-antisymmetric brackets. The symmetric part of this bracket usually has a very concrete form, written in terms of an inner product (in the $O(p,q)$ case) or its generalisation (in the exceptional case for $n<7$). Concretely, denoting the vector bundle describing the infinitesimal symmetries by $E$, there exists another vector bundle $N$, a vector bundle map $\langle\slot,\slot\rangle\colon E\otimes E\to N$, and a differential operator $\mc D$ from sections of $N$ to sections of $E$, such that for every section $u$ of $E$ we have
  \[[u,u]=\tfrac12\mc D\langle u,u\rangle.\]
  Although this condition was first encountered in the case of Courant algebroids (where $N$ is the trivial line bundle), it also holds in the other cases, notably in exceptional generalised geometry for $n<7$. It, however, ceases to hold in the cases $n\ge 7$.
  
  It is thus clear that in order to incorporate exceptional groups of higher rank, a somewhat radical shift of perspective is needed. Importantly, one notes that in the cases $n\geq8$ one expects the classical (Leibniz) algebroid framework to break down completely, due to the appearance of the dual graviton (c.f.\ the works~\cite{HS-E8,CR-E8} where various approaches are used to bypass this problem). The remaining case is the one of $n=7$, where the effect of the dual graviton is much milder. 
  
  In this work we abandon the inner product (or its generalisation) completely and study a new class of structures, which rely on a particular type of tensor, typically denoted $Y$ in the literature following~\cite{BCKT}. For this reason we call the resulting structures Y-algebroids. Crucially, we show that they have precisely the desired properties that are needed for the $n=7$ geometry.
  
  The fact that the algebra of symmetries can be concisely described via a bracket featuring $Y$ (as in~\eqref{eq:general}) has been known for some time~\cite{CSCW0,BCKT}. One of the points of the present work is to show that this expression can be regarded as a simple consequence of a more innocent-looking definition of a Y-algebroid. In addition, this approach gives a useful new handle on studying the appearance of fluxes, the phenomenon of Poisson--Lie U-duality, and the algebraic structure of maximally supersymmetric consistent truncations.

  The paper is organised as follows. We start with a short introduction into the general framework of Leibniz algebroids, and define and discuss basic concepts such as generalised Lie derivatives and connections. After a short linear algebra excursion we define Y-algebroids and their classes, and we show how they recover Lie algebroids, (generalisations of) Courant algebroids, and the previously defined exceptional algebroids for ranks $n<7$. We then study the $n=7$ case in more detail, and prove the basic structural results, a local classification theorem, and the algebraic characterisation of maximally supersymmetric consistent truncations (in terms of exceptional Manin pairs). We also touch lightly on the global classification issue. We finish with the discussion of the Poisson--Lie U-duality, look at some explicit examples, and take some steps connecting to exceptional complex geometry.
  
\section{Leibniz algebroids}
      We start by recalling the definition of a Leibniz algebroid. Although these structures are too general for our purposes here, they provide a useful starting point for defining and understanding notions such as generalised Lie derivatives and connections.
      
      A \emph{Leibniz algebroid} consists of a vector bundle $E\to M$, a vector bundle map $\rho\colon E\to TM$ (called the \emph{anchor}),\footnote{Vector bundle map means that at every point $m \in M$ we have a linear map on the fibres, $E_m\to T_mM$. This in particular induces a map (also denoted $\rho$) on the sections $\Gamma(E)\to \Gamma(TM)$, satisfying $\rho(fu)=f\rho(u)$ for any $f\in C^\infty(M)$ and $u\in\Gamma(E)$.} and an $\mb R$-bilinear bracket on the space of sections of $E$, \[[\slot,\slot]\colon \Gamma(E)\times\Gamma(E)\to\Gamma(E),\]
      such that for all $u,v,w\in\Gamma(E)$ and $f\in C^\infty(M)$ we have
      \[[u,[v,w]]=[[u,v],w]+[v,[u,w]],\qquad [u,fv]=f[u,v]+(\rho(u)f)v.\]
      The first condition is known as the \emph{Leibniz identity}.
      
      The simplest example is given by $E=TM$ with $\rho=\on{id}$ and $[\slot,\slot]$ the Lie bracket (commutator) of vector fields.
      It is often useful to regard Leibniz algebroids as a generalisation of this particular example.
      
      One simple consequence of the definition of a Leibniz algebroid is the fact that the anchor intertwines the brackets,
      \[\rho([u,v])=[\rho(u),\rho(v)].\]

      The operation $[u,\slot]$ extends naturally to a \emph{generalised Lie derivative} operator $L_u$ acting on tensors on $E$, that is, sections of \[E^*\otimes\dots \otimes E^*\otimes E\otimes\dots\otimes E.\] For instance, for $f\in C^\infty(M)$ and $\xi\in \Gamma(E^*)$ we have
    \[L_u f=\rho(u)f,\qquad \langle L_u \xi,v\ra=\rho(u)\la \xi,v\ra-\la \xi,[u,v]\ra,\]
    where $\la\slot,\slot\ra$ is the pairing between $E^*$ and $E$. A tensor $t$ on a Leibniz algebroid is called \emph{invariant} if $L_u t=0$ for all $u\in\Gamma(E)$. 
    
      A \emph{connection} on a Leibniz algebroid $E$ is an $\mb R$-bilinear map $\nabla\colon\Gamma(E)\times\Gamma(E)\to\Gamma(E)$, denoted $u\otimes v\mapsto \nabla_uv$, such that for all $u,v\in\Gamma(E)$ and $f\in C^\infty(M)$ we have
      \begin{equation}\label{eq:connection}
        \nabla_{fu}v=f\nabla_uv,\qquad \nabla_u (fv)=f\nabla_uv+(\rho(u)f)v.
      \end{equation}
    Again, $\nabla_u$ extends to an action on all tensors in $E$.
    Note that given a conventional vector bundle connection $D\colon\Gamma(TM)\times\Gamma(E)\to\Gamma(E)$ one can always define an associated algebroid connection $\nabla^{(D)}$ via $\nabla^{(D)}_uv:=D_{\rho(u)}v$ for all $u,v\in\Gamma(E)$. 
    A connection $\nabla$ is called \emph{$t$-compatible} if $\nabla t=0$.

      Finally, we will use $\rho^*\colon T^*M\to E^*$ to denote the transpose/dual map of the anchor, i.e.
        \[\la\rho^*(\alpha),u\ra=\la \alpha,\rho(u)\ra\qquad u\in \Gamma(E), \alpha\in\Omega^1(M).\]
  
  \section{Y-algebroids}
  Let us now proceed to the main protagonist of the present article, the \emph{Y-algebroid}. We first give the definition and then discuss some examples.

  \subsection{Definition}

  We start with a brief discussion of some linear algebra.
    We will be interested in studying pairs $(R,y)$ of a real vector space $R$ and a linear map \[y\colon R^*\otimes R\to R^*\otimes R,\] seen as a tensor on $R$. Two simple examples are given by $y=0$ and by taking the transpose of an endomorphism (w.r.t.\ some fixed inner product). It is shown below that these correspond to Lie and (a generalisation of) Courant algebroids, respectively.
    For any $(R,y)$ we define $\on{Aut}(y)$ as the subgroup of $GL(R)$ which preserves $y$.
    If we choose a basis $e_\alpha$ of $R$, we can denote the components of $y$ by
      \[y^{\alpha\beta}_{\,\gamma\delta}:=\la y(e^\alpha\otimes e_\gamma),e_\delta\otimes e^\beta\ra.\]
    Finally, we will say that a subspace $V\subset R$ is \emph{coisotropic} if
    \[y^{\;\alpha\beta}_{(\gamma\delta)}\xi_\alpha\zeta_\beta=0\quad \xi,\zeta\in V^\circ,\]
    where $V^\circ:=\{\xi\in R^* : \xi|_V=0\}\subset R^*$ is the annihilator of $V$.

\begin{defn}
  A \emph{Y-algebroid} is a Leibniz algebroid $(E,[\slot,\slot],\rho)$, together with an invariant tensor \[Y\colon E^*\otimes E\to E^*\otimes E,\] such that there exists a $Y$-compatible connection $\nabla$ satisfying 
  \begin{equation}\label{eq:bkt}
    [u,u]=Y(\nabla u)u,\qquad \forall u\in\Gamma(E).
  \end{equation}
\end{defn}

Note that, dropping the invariance and compatibility requirements in this definition, one recovers the \emph{anti-commutable Leibniz algebroid} of~\cite{DD} (for related previous works see \cite{GKP,JV}).
Notice also that the data necessary to define a Y-algebroid consists only of a Leibniz algebroid and a particular tensor $Y$, and there is typically no concrete or natural connection associated with such an algebroid.
      
We will be interested in cases where $Y$ ``is the same'' everywhere on the manifold. To make this formal, let $R$ be a vector space and $y\colon R^*\otimes R\to R^*\otimes R$ a fixed tensor. A Y-algebroid is said to be of \emph{class} $(R,y)$ if there is a linear isomorphism from each fiber to $R$ such that $Y$ is mapped to $y$. Note that looking at all possible identifications of the fibres of $E$ with $R$ which carry $Y$ to $y$ we obtain a principal $\on{Aut}(y)$-bundle.\footnote{Here we assume $\on{Aut}(y)\subset GL(R)$ is closed.}
        
\begin{prop}
For any Y-algebroid the kernel of $\rho$ is coisotropic at every point. Equivalently,
\begin{equation}\label{eq:seq}
  T^*M\otimes E\otimes E\xrightarrow{\check\rho} E\xrightarrow{\rho} TM\to 0, 
\end{equation}
with $\check\rho(\xi\otimes u\otimes v):=Y(\rho^*\xi\otimes u)v+Y(\rho^*\xi\otimes v)u$, is a chain complex.
\end{prop}
\begin{proof}
Applying $\rho$ to $[u,u]=Y(\nabla u)u$ we get $0=\rho(Y(\nabla u)u)$ for any $u\in\Gamma(E)$. In particular, for any $f\in C^\infty(M)$ we have
$0=\rho(Y(\nabla fu)fu)-f^2\rho(Y(\nabla u)u)=f\rho(Y(\rho^*d f\otimes u)u)$, which implies that $\rho\check\rho(df\otimes u\otimes u)=0$. Since any element of the fibre $T^*_mM$ can be extended to an exact 1-form, it follows that \eqref{eq:seq} is a chain complex. Looking at a fibre of a fixed point and noticing that
$\on{im}\rho^*=(\ker\rho)^\circ$,
we can rewrite $\rho(Y(\on{im}\rho^*\otimes u)u)=0$ as $\la Y((\ker\rho)^\circ\otimes u)u,(\ker\rho)^\circ\ra=0$, which finishes the proof.
\end{proof}
      Importantly, if \eqref{eq:seq} is an exact sequence, we say that the Y-algebroid is \emph{exact}. As seen below, it is precisely exact Y-algebroids that provide a natural setup for many string and M-theoretic applications. However, when dealing with Poisson--Lie-type dualities, it is crucial to regard exact Y-algebroids as a very special sector of a larger class of all Y-algebroids.\footnote{or at least of a larger class of \emph{transitive} Y-algebroids, i.e.\ those with $\rho$ surjective} This has been understood originally in the works of \v Severa~\cite{let,Scs}, which describe Poisson--Lie T-duality~\cite{KS} using (exact as well as non-exact) Courant algebroids.
      
    \subsection{Group data and \texorpdfstring{$G$}{G}-algebroids}
      The pairs $(R,y)$ that arise in string and M-theory come from group representations via the following construction:
      Suppose $R$ is a faithful representation of a Lie group $G$ and \[\pi\colon R^*\otimes R\to R^*\otimes R\] is an equivariant map with $\on{im}\pi=\mf g$ in the representation $R$. We then set $y:=\on{id}-\pi$. 
      
      Although by construction we always have $G\subset \on{Aut}(y)$, the two groups are in general not equal, and indeed they will not be for the cases of interest to use (see below). For example, for the $E_{7(7)} \times \mb R^+$ case,
      they differ slightly in their global structure:
      \ $E_{7(7)}\times \mb R^+$ vs $(\mb Z_2\ltimes E_{7(7)})\times \mb R^+$.
      
      On a related note, it is instructive to compare Y-algebroids with the previously defined $G$-algebroids~\cite{BHVW}. As shown below, these two notions essentially coincide in the case of exceptional groups of low rank (up to 6). Interestingly, not only do Y-algebroids allow one also to capture the case $n=7$, but their definition seems simpler. This is  due to the fact that it only relies on the notion of the tensor $Y$ (or $y$) --- one then recovers the group by studying maps preserving this tensor. The interaction between the bracket and the group is then ensured by the requirement that $Y$ is invariant.

\subsection{Examples}

\subsubsection*{(a) Lie algebroids}
Take $R=\mb R^n$ and $G=GL(n,\mb R)$, with $\pi=\on{id}$. We have $y=0$, $\on{Aut}(y)=GL(n,\mb R)$, and the corresponding algebroid has $[u,u]=0$. Y-algebroids of class $(\mb R^n,0)$ thus coincide with Lie algebroids of rank $n$. Furthermore, exact Y-algebroids of this class correspond to tangent Lie algebroids $E=TM$ with $\dim M=n$.

\subsubsection*{(b) Beyond Courant algebroids}
Take $R=\mb R^n$ to be the vector space with the inner product $\eta$ of signature $(p,n-p)$, forming the vector representation of $G=O(p,n-p)$, with $\pi(A)=A-A^T$. This implies $y(A)=A^T$, which we can write as \[y=\eta\otimes\eta^{-1}\in R^*\otimes R^*\otimes R\otimes R.\] In this case $\on{Aut}(y)$ is the group
\[\{A\in GL(n,\mb R)\mid A^T\!\eta A=\lambda\eta\text{ for some }\lambda\in\mb R\}.\]
Note that by Sylvester's law, $\lambda<0$ is only possible if $\eta$ has split signature. Thus
\[\on{Aut}(y)=\begin{cases}(\mb Z_2\ltimes O(\tfrac n2,\tfrac n2))\times \mb R^+ & \text{if $\eta$ has split signature,}\\O(p,n-p)\times \mb R^+ & \text{if not.}\end{cases}\]
To define the action of the $\mb Z_2$ factor, choose a basis $\{e_i^+,e_i^-\}$ such that $\eta=\sum_i (e^+_i\otimes e^+_i-e^-_i\otimes e^-_i)$. The group $\mb Z_2$ is generated by the flip $e^\pm_i\mapsto e^\mp_i$. Note that $A=-\on{id}$ is contained within $O(p,n-p)$. Finally, notice that if we see $y$ as a map \[\hat y\colon R\otimes R\to R\otimes R,\] then $\dim(\on{im}\hat y)=1$.
          
For any Y-algebroid of class $(\mb R^n,\eta\otimes\eta^{-1})$ we thus obtain a line bundle $N:=\on{im}\hat Y\subset E\otimes E$, and $\hat Y$ can be seen as a line-bundle-valued inner product $(\cdot,\cdot)\colon E\otimes E\to N$. One easily checks that, when acting on $\Gamma(N)$, the generalised Lie derivative satisfies $L_{fu}=fL_u$ for all $f\in C^\infty(M)$, $u\in\Gamma(E)$, i.e.\ it defines a (Leibniz algebroid) connection $\nabla^N$ on $N$. The invariance of $Y$ then means
\[\nabla^N_u( v,w)=([u,v],w)+(v,[u,w]).\]
On the other hand, the condition \eqref{eq:bkt} translates to
\[[u,u]=\tfrac12\on{Tr}\nabla^N(u,u),\]
where $\on{Tr}$ is understood as the contraction on $E^*\otimes N\subset E^*\otimes E\otimes E$.
          
Y-algebroids of this class are thus closely related (though not equivalent) to the structures introduced in~\cite{CLS},~\cite{B},~\cite{GS}, which generalise Courant algebroids~\cite{LWX}.

\subsubsection*{(c) \texorpdfstring{$E_n$}{En}-algebroids}
We now restrict our attention to the groups $E_{n(n)}$ and representations which appear in exceptional generalised geometry. For the sake of simplicity we focus on the case $n>3$, while for conceptual reasons (discussed later in the text) we consider the upper bound $n<8$. $E_{n(n)}$ are real Lie groups whose Lie algebras are split real forms of the corresponding complex simple Lie algebras.\footnote{or, in the case of lower rank, ``extrapolations'' thereof, via the Dynkin diagram pattern} We list them below, together with their Dynkin diagrams, and a chosen representation $R$, which is always fundamental, corresponding to the marked black node.
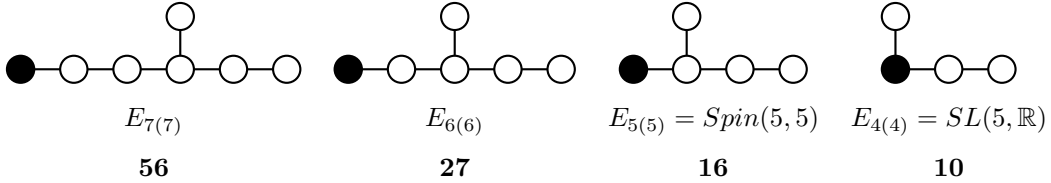
\begin{center}
  \begin{tabular}{cccc}
    $\begin{tikzpicture}[scale=.7, baseline=-0.6ex, thick]
      \draw (0,0) -- (5,0); \draw (3,1) -- (3,0); \node at (3,1) [proj0] {\phantom{+}}; \node at (0,0) [proj1] {\phantom{+}}; \node at (5,0) [proj0] {\phantom{+}};
      \foreach \x in {1,...,4}
      \node at (\x,0) [proj0] {\phantom{+}};
    \end{tikzpicture}$ &
    $\begin{tikzpicture}[scale=.7, baseline=-0.6ex, thick]
          \draw (0,0) -- (4,0); \draw (2,1) -- (2,0); \node at (2,1) [proj0] {\phantom{+}}; \node at (0,0) [proj1] {\phantom{+}}; \node at (4,0) [proj0] {\phantom{+}};
          \foreach \x in {1,...,3}
              \node at (\x,0) [proj0] {\phantom{+}};
          \end{tikzpicture}$ &
    $\begin{tikzpicture}[scale=.7, baseline=-0.6ex, thick]
          \draw (0,0) -- (3,0); \draw (1,1) -- (1,0); \node at (1,1) [proj0] {\phantom{+}}; \node at (0,0) [proj1] {\phantom{+}}; \node at (3,0) [proj0] {\phantom{+}};
          \foreach \x in {1,...,2}
              \node at (\x,0) [proj0] {\phantom{+}};
            \end{tikzpicture}$ &
    $\begin{tikzpicture}[scale=.7, baseline=-0.6ex, thick]
          \draw (0,0) -- (2,0); \draw (0,1) -- (0,0); \node at (0,1) [proj0] {\phantom{+}}; \node at (0,0) [proj1] {\phantom{+}}; \node at (2,0) [proj0] {\phantom{+}};
          \foreach \x in {1,...,1}
              \node at (\x,0) [proj0] {\phantom{+}};
          \end{tikzpicture}$ \\ [.2cm]
    $E_{7(7)}$ & $E_{6(6)}$ & $E_{5(5)}=Spin(5,5)$ & $E_{4(4)}=SL(5,\mb R)$
          \\ [.2cm]
    $\mathbf{56}$ & $\mathbf{27}$ & $\mathbf{16}$ & $\mathbf{10}$
  \end{tabular}
\end{center}
As is usual, we extend the group $E_{n(n)}$ to $G:=E_{n(n)}\times\mb R^+$, with the second factor acting on $R$ with weight 1.
In each case, we take the projection onto the Lie algebra~\cite{CSCW}
\begin{equation}\label{eq:pi}
  \pi=\alpha\,\pi_{\text{ad}}-\beta\,\pi_{\mb R},
\end{equation}
where $\pi_{\text{ad}}$ and $\pi_{\mb R}$ are the projections onto $\mf e_{n(n)}$ and $\mb R$, respectively, in the representation $R$, $\alpha$ is twice the Dynkin index of $R$, and $\beta=\frac{\dim R}{9-n}$. Specifically one has
\begin{center}
  \begin{tabular}{@{}ccccc@{}}
    \toprule
    $n$ & 7 & 6 & 5 & 4 \\
    \midrule
    $\alpha$ & 12 & 6 & 4 & 3 \\
    $\beta$ & 28 & 9 & 4 & 2 \\
    \bottomrule
  \end{tabular}
\end{center}
One can derive \eqref{eq:pi} from an even simpler expression \cite{BCKPS,CP}
\[\pi=t^a\otimes t_a+[1-(\Lambda,\Lambda)]1\otimes 1,\]
where $t_a$ are the generators of $\mf e_{n(n)}$, $\Lambda$ is the dominant integral weight associated to $R$, and we used the normalisation of the Killing form under which the long roots have length $2$.

An \emph{$E_n$-algebroid} is short for a Y-algebroid of class $(R,y)$ for these $R$ and $y$ (both depending on $n$).
          
Let us now show how this notion relates to the previously defined exceptional algebroids~\cite{BHVW} for $4\le n\le 6$. More concretely, we show that locally the two constructions coincide.\footnote{Note that there are local classification results~\cite{BHVW,BHVW2,HV} related to $G$-algebroids. Because of the correspondence between the two notions, we automatically obtain analogous local classifications of $E_n$-algebroids for $n<7$.}
First, let us show that
\begin{equation}\label{eq:autos}
  \mf{aut}(y)\cong \mf e_{n(n)}\oplus \mb R.
\end{equation}
In order to do this, note that any element $A\in \on{Aut}(y)\subset GL(R)$ preserves the image of $\pi =\on{id}-y$ in $R^*\otimes R$, i.e.\ the subspace $\mf{e}_{n(n)}\oplus\mb R\subset R^*\otimes R$. In other words,
\[A(\mf{e}_{n(n)}\oplus\mb R) A^{-1}\subset \mf{e}_{n(n)}\oplus\mb R.\]
Since $A[x,y]A^{-1}=[AxA^{-1},AyA^{-1}]$ and conjugation acts trivially on the second summand of $\mf{e}_{n(n)}\oplus\mb R$, the assignment $A\mapsto A(\slot) A^{-1}$ gives a map to Lie algebra automorphisms,
\[\on{Aut}(y)\to \on{Aut}(\mf e_{n(n)}).\]
Schur's lemma implies that the kernel of this map is $\mb R^*$. Since $\on{Lie}(\on{Aut}(\mf e_{n(n)}))\cong \mf e_{n(n)}$ (due to simplicity), we obtain the identification \eqref{eq:autos}. Notice however that the relation between the Lie groups $\on{Aut}(y)$ and $E_{n(n)}\times\mb R^+$ is more complicated due to the presence of outer automorphisms (note that the Dynkin diagrams for $n\in \{4,5,6\}$ possess a $\mb Z_2$ symmetry).
          
Second, for $n<7$, $y$ is symmetric both in its upper and lower indices, similar to the case of conformal Courant algebroids. Furthermore, defining $\hat y\colon R\otimes R\to R\otimes R$ to be its partial dual, we have that $R':=\on{im}\hat y\subset R\otimes R$ is a particular subrepresentation:
\begin{center}
  \begin{tabular}{@{}cccc@{}}
    \toprule
    $n$ & 6 & 5 & 4\\
    \midrule
    $R'$ & $\mathbf{27}'$ & $\mathbf{10}$ & $\mathbf{5}'$ \\
    \bottomrule
  \end{tabular}
\end{center}
If now $E$ is an $E_n$-algebroid for $n<7$, and we define the subbundles
\[N:=\on{im}\hat Y\subset E\otimes E,\qquad N^*:=\on{im}\hat Y^*\subset E^*\otimes E^*,\]
we get symmetric maps $E\otimes E\to N$ and $E^*\otimes E^*\to N^*$ (or equivalently $N\to E\otimes E$). Then
\[Y(\nabla u)u=\tfrac12\on{Tr}\nabla(\hat Y(u\otimes u)),\]
where again $\on{Tr}$ denotes the contraction on $E^*\otimes N\subset E^*\otimes E\otimes E$. We thus obtain an exceptional algebroid~\cite{BHVW}, with $\mc Dn=\on{Tr}\nabla n$.\footnote{The principal bundle, which is required in the definition of an exceptional algebroid, corresponds locally to the $\on{Aut}(y)$-principal bundle which we obtain from the Y-algebroid $E$. Furthermore, $E$ can be seen as the associated bundle for the representation $R$.}
Conversely, starting with an exceptional algebroid we define $\hat Y$ to be the composition $E\otimes E\to N\to E\otimes E$. Similarly, it is not difficult to see that one can always extend $\mc D$ to a $Y$-compatible connection $\nabla$ such that $\on{Tr}\nabla n=\mc Dn$. This establishes the correspondence between the two notions.

\subsubsection*{(d) Other generalised geometries}
There are also further interesting classes of geometries. Two of those arise as a straightforward generalisation of the cases $E_{4(4)}\times\mb R^+$ and $E_{5(5)}\times\mb R^+$, and are given by the following diagrams:
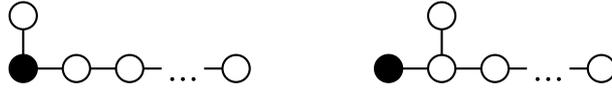
\begin{center}
  $\begin{tikzpicture}[scale=.7, baseline=-0.6ex, thick]
    \draw (0,0) -- (2.6,0); \draw (3.4,0) -- (4,0); \draw (0,1) -- (0,0); \node at (0,1) [proj0] {\phantom{+}}; \node at (0,0) [proj1] {\phantom{+}}; \node at (2,0) [proj0] {\phantom{+}}; \node at (4,0) [proj0] {\phantom{+}};
    \foreach \x in {1,...,1}
    \node at (\x,0) [proj0] {\phantom{+}};
    \filldraw [black] (2.8,-.2) circle (.5pt); \filldraw [black] (3,-.2) circle (.5pt); \filldraw [black] (3.2,-.2) circle (.5pt);
  \end{tikzpicture}$\hspace{1.5cm}
  $\begin{tikzpicture}[scale=.7, baseline=-0.6ex, thick]
    \draw (0,0) -- (2.6,0); \draw (3.4,0) -- (4,0); \draw (1,1) -- (1,0); \node at (1,1) [proj0] {\phantom{+}}; \node at (0,0) [proj1] {\phantom{+}}; \node at (2,0) [proj0] {\phantom{+}}; \node at (4,0) [proj0] {\phantom{+}};
    \foreach \x in {1,...,1}
    \node at (\x,0) [proj0] {\phantom{+}};
    \filldraw [black] (2.8,-.2) circle (.5pt); \filldraw [black] (3,-.2) circle (.5pt); \filldraw [black] (3.2,-.2) circle (.5pt);
  \end{tikzpicture}$
\end{center}
The first one is given by the two-form representation of $SL(n,\mb R)\times\mb R^+$~\cite{PS,LSCW} and plays an important role in the description of Freund--Rubin compactifications and their consistent truncations~\cite{LSCW,HS,EGM}.
The second one, introduced and studied in~\cite{SC}, corresponds to the spinor representation of $Spin(n,n)\times\mb R^+$.

Another intriguing class of geometries \cite{BCKPSS} with non-split real forms of Lie groups arises from the study of magical supergravities \cite{GST,GST2,GSS}. This in particular includes the so-called \emph{type $\mf e_7$} setups, describing reductions to 4 dimensions. The relevant groups and representations in this case are
\begin{equation}\label{eq:non_split}
    (Sp(6,\mb R)\times \mb R^+,\r{14}),\quad (SU(3,3)\times \mb R^+,\r{20}),\quad (SO^*(12)\times \mb R^+,\r{32}),\quad (E_{7(-25)}\times \mb R^+,\r{56}).
\end{equation}
Similarly to the above $E_{7(7)}\times\mb R^+$-case, the tensor $y$ is not symmetric in the upper/lower indices. Hence the corresponding geometries do not admit a $G$-algebroid description. However, just as the $E_{7(7)}\times\mb R^+$ case, they fit well in the framework of Y-algebroids.

\section{The \texorpdfstring{$n=7$}{n=7} case}
    Let us now turn to the main focus of the paper, namely the application of Y-algebroids in the case of $E_{7(7)}$.

\subsection{General M-exact \texorpdfstring{$E_7$}{E7}-algebroids}
\label{subsec:M-exact}    

Let us start with some algebraic preliminaries, used to define the group $E_{7(7)}$ and the automorphisms of the $y$ tensor. Set $R:=\w2\mb R^8\oplus \w2\!(\mb R^8)^*$. 
One can then define the following symplectic and quartic forms on $R$:
\[
  \omega((x,y),(x',y'))=\on{Tr} xy'-\on{Tr}x'y, \qquad
  q(x,y)=\on{Tr}(xy)^2-\tfrac14(\on{Tr}xy)^2+4(\on{Pf}x+\on{Pf}y).
\]
The group $E_{7(7)}$ is defined as the subgroup of $GL(R)$ which preserves $\omega$ and $q$.

The $y$ tensor takes the form \cite{BCKT}:
\begin{equation}\label{eq:y}
  y^{\alpha\beta}_{\,\gamma\delta}=12\,\omega^{\alpha\epsilon}\omega^{\beta\zeta}q_{\epsilon\zeta\gamma\delta}+\delta^{(\alpha}_\gamma \delta^{\beta)}_\delta+\tfrac12(\omega^{-1})^{\alpha\beta}\omega_{\gamma\delta}.
\end{equation}
The subgroup $\on{Aut}(y)\subset GL(R)$ which preserves $y$ is, however, slightly larger than $E_{7(7)}\times\mb R^+$, and can be found via the following argument. First, note that $\on{Aut}(y)$ has to preserve separately the symmetric and antisymmetric (in pairs of lower and upper indices) parts of $y$. The antisymmetric part is proportional to $\omega\otimes\omega^{-1}$, which implies that $\on{Aut}(y)$ is a subgroup of conformal symplectic transformations
\[
  \{A\in GL(R)\mid A^T \omega A=\lambda \omega \text{ for some }\lambda\in\mb R\}\cong (\mb Z_2\ltimes Sp(56,\mb R))\times\mb R^+. 
\]
The $\mb Z_2$ factor here can be identified with the subgroup of $GL(R)$ generated by the flip \[R\ni (x,y)\mapsto (-x,y),\]
or equivalently with the subgroup generated by any other element of $GL(R)$ related to this flip by an $Sp(56,\mb R)$-conjugation. (We will see in a moment how to interpret the flip.) For the symmetric part of $y$, the second term in \eqref{eq:y} is automatically invariant under the entire $GL(R)$, while the first term is invariant w.r.t.\ the scaling, the flip, and the part of $Sp(56,\mb R)$ preserving $q$. Hence
      \[\on{Aut}(y)=(\mb Z_2\ltimes E_{7(7)})\times \mb R^+,\qquad \mf{aut}(y)=\mf{e}_{7(7)}\oplus\mb R.\]

Let us now set some basic nomenclature which we will use in the case of $E_7$-algebroids. First, a coisotropic subspace of $R$ is said to be \emph{of type M} when it has codimension 7. Using the same approach as in \cite{BHVW}, one can show that all subspaces of type M are related by an action of $E_{7(7)}$, i.e.\ they form a single orbit. Explicitly, using the formulas from \cite{CSCW}, we have the following decompositions under the subgroup $GL^+(7,\mb R)\subset E_{7(7)}\times\mb R^+$:
\begin{align}
  \label{eq:repdep}
  R &=T\oplus\w2T^*\oplus\w5T^*\oplus(T^*\otimes\w7T^*) , \\
  \label{eq:deco}
  \mf{e}_{7(7)}\oplus\mb R &=\mb R\oplus (T^*\otimes T)\oplus \w3T^*\oplus\w6T^*\oplus\w3T\oplus\w6T ,
\end{align}
where $T:=\mb R^7$. The subspace $\w2T^*\oplus\w5T^*\oplus(T^*\otimes\w7T^*)\subset R$ is then of type M. Note that of the $GL(7,\mb R)$ action on $T$, only the set of elements with positive determinant is a subgroup of $E_{7(7)}\times \mb R^+$. However, the $\mb Z_2$ flip symmetry $(x,y)\to(-x,y)$ acts as $-\on{id}$ on $T$ and $\w5T^*$ and trivially on the other summands in~\eqref{eq:repdep}, and so it corresponds to the action of $-\on{id}\in GL(7,\mb R)$. Thus $\on{Aut}(y)$ is precisely the extension of $E_{7(7)}\times \mb R^+$ required to include the full $GL(7,\mb R)$ action on $R$. 
        
An $E_7$-algebroid is called \emph{M-exact} if any one of the following equivalent conditions are satisfied:
\begin{itemize}
\item the sequence \eqref{eq:seq} is exact and the base is 7-dimensional
\item $\rho$ is surjective and $\ker\rho$ is of type M everywhere
\item $\rho$ is surjective and the base is 7-dimensional.
\end{itemize}

Finally, note the following important fact, which holds for subspaces $V\subset R$ in the $n=7$ case (and trivially also for $n<7$):
\begin{equation}\label{eq:implies}
  y^{\,\alpha\beta}_{(\gamma\delta)}\xi_\alpha\zeta_\beta=0\quad \forall\xi,\zeta\in V^\circ
           \qquad \implies \qquad y^{\alpha\beta}_{\,\gamma\delta}\xi_\alpha\zeta_\beta=0\quad \forall\xi,\zeta\in V^\circ.
\end{equation}
We will use this property in Subsection~\ref{subsec:cons} when performing the algebraic characterisation of consistent truncations.
      
\subsubsection*{Example: Exceptional tangent bundle}

Motivated by \eqref{eq:repdep}, for any 7-dimensional manifold $M$ we can define the \emph{exceptional tangent bundle}
\[\mb TM:=TM\oplus\w2T^*M\oplus\w5T^*M\oplus(T^*M\otimes\w7T^*M),\]
with $\rho$ being the projection onto $TM$, the bracket
\begin{align*}
  [X+\sigma_2+\sigma_5+\tau, X' +\sigma_2'&+\sigma_5'+\tau']_0 \\
  &= \mc L_X X'
  + (\mc L_X \sigma_2'-i_{X'}d\sigma_2) 
  + (\mc L_X \sigma_5'-i_{X'}d\sigma_5-\sigma_2'\wedge d\sigma_2) \\
  & \qquad + (\mc L_X \tau'-j\sigma_5'\wedge d\sigma_2-j\sigma_2'\wedge d\sigma_5),
\end{align*}
and the $Y$-tensor given by \eqref{eq:pi} (or explicitly by \eqref{eq:y}). Here $j$ is the map \[\w pT^*\otimes \w{8-p}T^*\to T^*\otimes\w7T^*,\] defined by
      \[(j\alpha\wedge\beta)(X)=(i_X\alpha)\wedge \beta\in\w7T^*,\qquad X\in T.\]
      
In order to see that this defines a Y-algebroid, we note that the above bracket can be written more concisely as follows. Let $D$ be any torsion-free conventional connection on $TM$. This naturally extends to a conventional connection on $\mb TM$ and hence defines an associated algebroid connection $\nabla^{(D)}$. By definition $DY=0$ and hence $\nabla^{(D)}Y=0$. Using the fact that, since $D$ is torsion-free, $\mathcal{L}_XX'=D_XX'-D_{X'}X$,  $d\sigma_2=d_D\sigma_2$, etc., one can rewrite \cite{CSCW}
       \begin{equation}\label{eq:general}
        [u,v]_0=\nabla^{(D)}_uv-(1-Y)(\nabla^{(D)}u)v
        =\nabla^{(D)}_uv-\nabla^{(D)}_vu+Y(\nabla^{(D)}u)v, 
      \end{equation}
independent of the choice of $D$. We then immediately see that the last axiom for a Y-algebroid is satisfied for the connection $\nabla^{(D)}$.           
      
\subsubsection*{Classification}\label{subsec:class}
Two M-exact $E_7$-algebroids $E$ and $E'$ over $M$ are equivalent if there exists a vector bundle isomorphism $\phi\colon E\to E'$ that is compatible with the brackets, has $\phi_*(Y)=Y'$ and commutes with the anchor maps. We can then consider classifying general M-exact $E_7$-algebroids up to equivalence, both locally and globally.   

Starting with the local problem, suppose that $E\to M$ is a general M-exact $E_7$-algebroid. One can check --- again by repeating the calculation from \cite{BHVW} --- that locally (over a contractible open set $\mc U$) there exists a vector bundle isomorphism $E\to \mb T\,\mc U$ which maps the $Y$-tensor on $E$ to that on $\mb T\,\mc U$ and commutes with the anchors, i.e.\ it makes the following diagram commute
      \[\begin{tikzcd}[column sep=small] E \arrow[rr] \arrow[dr, swap, "\rho"] & & \mb T\,\mc U \arrow[dl,"\rho"]\\ & T\mc U &\end{tikzcd}\]
      However, such a map is not unique, since we can compose it with any vector bundle automorphism of $\mb T\,\mc U$ which preserves $Y$ and $\rho$. Such automorphisms are parametrised by maps $\mc U\to S$, where $S$ is the Lie group with Lie algebra \[\mf s=\mb R'\oplus \w3T^*\oplus \w6T^*\subset \mf e_{7(7)}\oplus\mb R,\]
      with $\mb R'$ being the particular combination of the two $\mb R$'s in \eqref{eq:deco},
      \[\mb R'\subset\mb R\oplus\mb R\subset \mb R\oplus (T^*\otimes T)\subset \mf e_{7(7)}\oplus\mb R,\]
      that fixes elements of $T\subset R$. In other words, infinitesimal automorphisms are given by an arbitrary function, 3-form, and a 6-form.
      
      Using one such local identification $E\cong \mb T\,\mc U$, we see that any M-exact $E_7$-algebroid can be locally written as
      \[T\,\mc U\oplus\w2T^*\mc U\oplus\w5T^*\mc U\oplus(T^*\mc U\otimes\w7T^*\mc U),\]
      with $Y$ and $\rho$ having the same form as on $\mb T\,\mc U$. What remains is to determine the possible form of the bracket on $\mb T\,\mc U$. The definition of a Y-algebroid requires the following behaviour of the bracket $[\slot,\slot]$ under the multiplication of the arguments by a function:
      \[[u,fv]-f[u,v]=(\rho(u)f)v,\qquad [fu,v]-f[u,v]=(Y-1)(\rho^*df\otimes u)v.\]
      We know that $\mb T\,\mc U$ can be endowed with the bracket $[\slot,\slot]_0$ and connection $\nabla^{(D)}$, associated to a torsion-free conventional connection $D$ on $\mc U$, as discussed in the preceding example. Since the bracket $[\slot,\slot]_0$ on $\mb T\,\mc U$ does satisfy these properties, the most general bracket on $\mb T \,\mc U$ has to differ from it by a purely tensorial quantity, i.e.
      \[[u,v]=[u,v]_0+A(u)v,\]
      for some vector bundle map $A\colon \mb T\,\mc U\to (\mf e_{7(7)}\oplus\mb R)\times \mc U$.
      In addition, since the same anchor on $\mb T\,\mc U$ intertwines both $[\slot,\slot]$ and $[\slot,\slot]_0$ with vector fields on $\,\mc U$, we have that $\rho(A(u)v)=0$, i.e.\
      \[A\colon \mb T\,\mc U\to \mf s\times \mc U.\]
      Similarly, any connection can be written as $\nabla_u=\nabla^{(D)}_u+B(u)$, for some
      \[B\colon \mb T\,\mc U\to (\mf e_{7(7)}\oplus\mb R)\times \mc U.\]
      The last axiom in the definition of a Y-algebroid now imposes the condition that there exists such a $B$ solving the tensor equation
      \[A(u)u=Y(B(\slot)u)u,\qquad \forall u\in\mb T\,\mc U.\]

      An explicit calculation reveals that the only $A$ admitting a solution have the form
      \begin{equation*}
        A(X+\sigma_2+\sigma_5+\tau)=i_XF_1+(i_XF_4-F_1\wedge\sigma_2)+(i_XF_7-F_4\wedge\sigma_2-2F_1\wedge\sigma_5),
      \end{equation*}
      for some $F_1\in\Omega^1(\mc U)$, $F_4\in\Omega^4(\mc U)$, and $F_7\in\Omega^7(\mc U)$.
      Finally, imposing the Leibniz identity for the bracket produces the Bianchi identities
      \begin{equation}\label{eq:bianchi}
        dF_1=0,\qquad dF_4+F_1\wedge F_4=0.
      \end{equation}

      Note that there is no Bianchi for $F_7$ as we are on a 7-dimensional manifold. Physically, the forms $F_1$, $F_4$, and $F_7$ can be associated with the field strengths of the fields present in the reduction of M-theory, namely the scalar warp factor, the 3-form potential, and the dual 6-form potential, respectively.
      
      We have now determined that the most general bracket has (locally) a very specific form, parametrised only by a set of fluxes $F_1$, $F_4$, $F_7$.
      Changing the identification $E\cong \mb T\,\mc U$ by a map $\mc U\to S$ preserves this form, up to a change in the fluxes. At the infinitesimal level, acting with $c_0\in\Omega^0(\mc U)$, $c_3\in\Omega^3(\mc U)$, and $c_6\in\Omega^6(\mc U)$ results in
\begin{align*}
  c_0: && \delta F_1 &=-dc_0, & \delta F_4&=c_0F_4, & \delta F_7&=2c_0F_7, \\
  c_3: && \delta F_1 &=0, & \delta F_4&=-dc_3-F_1\wedge c_3, &
                 \delta F_7&=F_4\wedge c_3, \\
  c_6: && \delta F_1 &=0, & \delta F_4&=0, & \delta F_7&=-dc_6-2F_1\wedge c_6,
\end{align*}
      respectively. In particular, any set of fluxes satisfying \eqref{eq:bianchi} can be locally removed by these transformations. This establishes the following result.
      \begin{thm}\label{thm:class}
        Every M-exact $E_7$-algebroid locally has the form of the untwisted exceptional tangent bundle.
      \end{thm}
      
We can extend this to also characterise the global classification of M-exact $E_7$-algebroids. From the above transformations it follows that the automorphisms (i.e.\ bundle isomorphisms preserving the brackets, anchors, and $Y$ tensors) of the untwisted exceptional tangent bundle $\mb T\,\mc U$ form a group $\mc G(\mc U)$ whose Lie algebra is given by elements of the form 
      \[\{c_0+c_3+c_6\in\Omega^0(\mc U)\oplus\Omega^3(\mc U)\oplus\Omega^6(\mc U)\mid dc_0=dc_3=dc_6=0\},\]
with the bracket  \[[c_0+c_3+c_6,c_0'+c_3'+c_6']=(c_0c_3'-c_0'c_3)+(2c_0c_6'-2c_0'c_6-c_3\wedge c_3').\]
  Note that this holds for any open set $\mc U\subset M$, not just for contractible ones. We thus obtain a sheaf $\mc G$ of non-abelian groups, with $\mc G(\mc U)$ being the group of automorphisms of $\mb T\,\mc U$.
  
Theorem~\ref{thm:class} implies that for any M-exact $E_7$-algebroid $E\to M$ we can find an open cover $\{\mc U_\alpha\}$ of $M$ such that on each $\mc U_\alpha$ have an isomorphism $\varphi_\alpha$ between the untwisted exceptional tangent bundle $\mb T\,\mc U_\alpha$ and $E$. On overlaps $\mc U_\alpha\cap \,\mc U_\beta$ we thus get automorphisms
      \[g_{\alpha\beta}:=\varphi_\beta^{-1}\circ \varphi_\alpha\]
      of the untwisted exceptional tangent bundle $\mb T(\mc U_\alpha\cap \,\mc U_\beta)$, i.e.\ $g_{\alpha\beta}\in\mc G(\mc U_\alpha\cap\,\mc U_\beta)$. On triple overlaps $\mc U_\alpha\cap \,\mc U_\beta\cap \,\mc U_\gamma$ this satisfies
      \[g_{\alpha\gamma}=g_{\beta\gamma}g_{\alpha\beta},\]
      and so we obtain a \v Cech 1-cocycle with values in the sheaf $\mc G$. Changing the identifications on $\mc U_\alpha$ from $\varphi_\alpha$ to $\varphi_\alpha h_\alpha$, the cocycle $\{g_{\alpha\beta}\}$ changes to $\{h_\beta^{-1}\!g_{\alpha\beta}h_\alpha\}$. The M-exact $E_7$-algebroids over $M$ are thus classified by the \v Cech cohomology\footnote{Note that when constructing an algebroid by gluing together smaller pieces, we can use a partition of unity to get a suitable globally defined connection.} \[\check H^1(M,\mc G).\]
      We leave the question of rewriting this cohomology in terms of more standard geometric structures to a future work.
      
    \subsection{Consistent truncations and Poissson--Lie U-duality}\label{subsec:cons}
      It has been shown in \cite{LSCW,HS} that consistent truncations to maximally supersymmetric $(11-n)$-dimensio\-nal theories correspond to \emph{Leibniz parallelisations}, i.e.\ global frames $e_\alpha$ of the (possibly twisted) exceptional tangent bundle with $Y^{\alpha\beta}_{\;\gamma\delta}$ and $c^\gamma_{\alpha\beta}$ constant, where the latter are the structure coefficients defined by $[e_\alpha,e_\beta]=c^\gamma_{\alpha\beta}e_\gamma$ and physically correspond to the embedding tensor of the lower-dimensional theory. Let us reformulate this in the present language, for the case of an $E_7$-algebroid.

\subsubsection*{Exceptional Manin pairs}
      
      Suppose we start with a Leibniz parallelisation of an M-exact $E_7$-algebroid $E$ over a compact connected base $M$, with $\nabla$ satisfying \eqref{eq:bkt}. This induces a trivialisation of the bundle, \[E\cong M\times\mf a,\] with the vector space $\mf a$ inheriting a well-defined Leibniz bracket and an invariant tensor $Y$; the basis of $\mf a$ gives the global frame $e_\alpha$ of $E$. Although $\nabla$ does not necessarily descend to $\mf a$, we can always replace it by one which does, due to the following argument.
       First, note that the trivialisation gives $\Gamma(E)\cong C^\infty(M)\otimes\mf a$, and consequently induces a second connection $\hat \nabla_uv:=\rho(u)v$ satisfying $\hat \nabla Y=0$. We set $B:=\nabla-\hat\nabla$. Seeing $x\in\mf a$ as a constant section of $E$ (i.e.\ $\hat\nabla x=0$), we then have
      \[[x,x]=Y(B(\slot)x)x.\]
      Note that the LHS is constant by assumption. Hence, for any two points $m,m'\in M$ we have \[Y((B(\slot)_m-B(\slot)_{m'})x)x=0\] for all $x\in\mf a$. Consequently, we can fix some $m\in M$ and replace the connection $\nabla=\hat\nabla+B(\slot)$ by $\hat\nabla+B(\slot)_m$, which still satisfies \eqref{eq:bkt}, but also maps constant sections to constant sections.
      
      The Leibniz algebra $\mf a$ thus inherits a connection (given simply by $B_m$) and becomes an $E_7$-algebroid over a point base (with $\rho=0$),
      \[\mf a\to \{*\}.\]
      Although $\mf a$ is not a Lie algebra in general, we can form a Lie algebra $\mf a_{\text{Lie}}$ out of it by modding out by the ideal $\mf i$ spanned by elements of the form $[u,u]$ for $u\in \mf a$,
      \[\mf a_{\text{Lie}}:=\mf a/\mf i.\]
            
      The anchor map of $E$ translates to an action of $\mf a$ on $M$, i.e.\ to a homomorphism $\mf a\to\Gamma(TM)$. Since the anchor is surjective (because the algebroid is M-exact) and $\mf i$ is in its kernel, the homomorphism descends to a transitive action of $\mf a_{\text{Lie}}$ on $M$. Using the compactness of $M$ we then have that $M\cong A_{\text{Lie}}/B_{\text{Lie}}$, where $A_{\text{Lie}}$ is the unique connected and simply connected Lie group integrating $\mf a_{\text{Lie}}$, and $B_{\text{Lie}}\subset A_{\text{Lie}}$ is a subgroup corresponding to a Lie subalgebra $\mf b_{\text{Lie}}\subset \mf a_{\text{Lie}}$. Finally, we take $\mf b\subset \mf a$ to be the preimage of $\mf b_{\text{Lie}}$ under $\mf a\to \mf a/\mf i$. Note that $\mf b$ is a subalgebra containing $\mf i$. Looking at the coset of the identity, $[1]\in A_{\text{Lie}}/B_{\text{Lie}}\cong M$, we see that the kernel of the anchor coincides precisely with $\mf b$, and so $\mf b$ has to be a type M subspace.

Summarising, any M-exact Leibniz parallelisable $E_7$-algebroid (over a compact connected base) produces a pair of \begin{itemize}
        \item an $E_7$-algebroid $\mf a\to \{*\}$
        \item a type M subalgebra $\mf b\subset \mf a$ containing the ideal $\mf i$.
      \end{itemize}
      We call such data an \emph{exceptional Manin pair}. Furthermore, it follows from the axioms of a Y-algebroid that the original M-exact $E_7$-algebroid is uniquely determined in terms of this data.
      
      Let us now show the converse, namely that any exceptional Manin pair comes from an M-exact $E_7$-algebroid. Starting from $\mf b\subset\mf a$, we take a corresponding pair of groups $B_{\text{Lie}}\subset A_{\text{Lie}}$ and construct the trivial vector bundle $\mf a\times A_{\text{Lie}}/B_{\text{Lie}}\to A_{\text{Lie}}/B_{\text{Lie}}$.\footnote{There is a small subtlety here in that $B_{\text{Lie}}$ needs to be a closed subgroup, in order to get a well-behaved quotient. Also, for a given pair of $\mf b\subset \mf a$ there exist in general multiple possible subgroups $B_{\text{Lie}}\subset A_{\text{Lie}}$, differing by their number of connected components. We ignore these technicalities.} We equip this bundle with a $Y$-tensor coming from $\mf a$, and with an anchor map coming from the action of $\mf a$ on $A_{\text{Lie}}/B_{\text{Lie}}$. Finally, understanding the sections of the vector bundle as functions on $M$ valued in $\mf a$, we define the bracket by
      \[[u,v]=\rho(u)v-\rho(v)u+Y(\rho^*du)v+[u,v]_\mf a,\]
      where the last term is the bracket on $\mf a$.
      One easily sees that this structure in general satisfies all the axioms of an M-exact $E_7$-algebroid apart from the Leibniz identity. Furthermore, one easily checks that for arbitrary (not necessarily constant) sections $u,v$ we have
      \[\rho([u,v])-[\rho(u),\rho(v)]=\rho(Y(\rho^*du)v),\]
      which vanishes due to the coisotropy of $\mf b$ and the property \eqref{eq:implies} of $Y$. We now use a trick, employed in \cite{BHVW}.
      
      We start by noting that in Subsection \ref{subsec:class}, when performing the classification of M-exact $E_7$-algebroids, we have only used the Leibniz identity at the very end, when deriving the Bianchi identities. Up to that point, we only used the other Y-algebroid axioms, together with the fact that $\rho([u,v])=[\rho(u),\rho(v)]$. Since the structure that we just constructed over $A_{\text{Lie}}/B_{\text{Lie}}$ does satisfy these properties, we know that the bracket locally needs to have the form of the exceptional tangent bundle twisted by some $F_1, F_4, F_7$ --- it is however not clear in general whether these fluxes satisfy the relevant Bianchi identities. Still, it is straightforward to check that even when the fluxes do not satisfy the Bianchi identities, the corresponding \emph{Jacobiator}
      \[J(u,v,w):=[u,[v,w]]-[[u,v],w]-[v,[u,w]]\]
      is always a tensor. Now since $\mf a$ is a Y-algebroid, its corresponding Jacobiator vanishes. This implies that the Jacobiator necessarily vanishes on the constant sections of $\mf a\times A_{\text{Lie}}/B_{\text{Lie}}$, and hence --- due to tensoriality --- has to be identically zero, implying that we indeed obtain a Y-algebroid. We have thus established both directions of:
      \begin{thm}[Classification of maximally supersymmetric consistent truncations]
        Leibniz parallelisations of M-exact $E_7$-algebroids correspond to exceptional Manin pairs.\footnote{The correspondence is one-to-one, up to the technicalities mentioned in the previous footnote.}
      \end{thm}
      \noindent 
      As for the classification theorem~\cite{BHVW,BHVW2,HV} for $n\leq6$, this is consistent with an earlier result of Inverso~\cite{I} derived using different techniques. We also note that the explicit structure of the corresponding Leibniz algebras and global frames, including the $n=7$ case, was recently expanded upon in~\cite{HasSak}. 

\subsubsection*{Poisson--Lie U-duality}
      Sometimes it is possible to find distinct exceptional Manin pairs which share the same $\mf a$. The corresponding Leibniz parallelisable spaces are then \emph{Poisson--Lie U-dual} \cite{Sakatani,MT,MST}, for the following reason.
      
      Recall that a \emph{generalised metric} is a reduction of the structure group from $E_{7(7)}\times\mb R^+$ to its maximally compact subgroup $SU(8)/\mb Z_2$. On an M-exact $E_7$-algebroid, such a reduction parametrises the bosonic field content of the 11-dimensional supergravity restricted to a 7-dimensional warped compactification manifold.      
      It is then possible to generalise the standard construction of Levi-Civita connections and define a generalisation of the Ricci tensor, whose vanishing corresponds to the equations of motion of the supergravity~\cite{CSCW}.
      
      Suppose now that we have two Manin pairs $(\mf a,\mf b)$, $(\mf a,\mf b')$, leading to Leibniz parallelisations of M-exact $E_7$-algebroids over spaces $A_{\text{Lie}}/B_{\text{Lie}}$ and $A_{\text{Lie}}/B'_{\text{Lie}}$, respectively. Equipping $\mf a$ with a generalised metric, we obtain induced generalised metrics on this pair of algebroids. Because of the relation of the brackets on these algebroids with the one on $\mf a$, it is easy to see (c.f.\ \cite{BHVW}) that the equations of motion hold on $A_{\text{Lie}}/B_{\text{Lie}}$ if and only if they hold on $A_{\text{Lie}}/B'_{\text{Lie}}$.
    \subsection{Examples}
      We now briefly discuss some standard examples of Leibniz parallelisations and Poisson--Lie U-duality in this language. 
      \subsubsection*{(a) Tori and U-duality}
        The simplest case comes from taking $\mf a=\mb R^{56}$ to be abelian and $\mf b$ an arbitrary type M subspace. This gives rise to a Leibniz parallelisation of the torus $T^7$. Note that different type M subspaces are obtained via the action of $E_{7(7)}\times\mb R^+$; Poisson--Lie U-duality in this case reduces to ordinary U-duality.
      \subsubsection*{(b) Groups and generalised Yang--Baxter deformations}
        Let $K$ be a 7-dimensional Lie group. Then the exceptional tangent bundle $\mb TK$ has a natural Leibniz parallelisation, induced by left invariant tensor fields. The corresponding exceptional Manin pair is given by
        \[\mf a=\mf k\oplus\w2\mf k^*\oplus\w5 \mf k^*\oplus (\mf k^*\otimes\w7\mf k^*),\qquad \mf b=\w2\mf k^*\oplus\w5 \mf k^*\oplus (\mf k^*\otimes\w7\mf k^*),\]
        with the bracket obtained from $[\slot,\slot]_0$ by making the replacement
        \[\text{Lie derivative} \to \text{(co)adjoint action},\;\; \text{de Rham differential}\to\text{Chevalley--Eilenberg differential}.\]
        This bracket can be twisted by any set of invariant forms $F_1$, $F_4$, $F_7$ satisfying the Bianchi identities.
        One can also try to study Poisson--Lie U-dual setups. In particular, requiring that the new $\mf b'$ be transverse to $\mf k\subset \mf a$ corresponds to generalised Yang--Baxter deformations, c.f.\ \cite{BHVW2}.
        \subsubsection*{(c)
          7-sphere}
        Take $\mf a:=\mf{so}(8)\oplus\mathbf{28}$ and $\mf b:=\mf{so}(7)\oplus\mathbf{28}$. The bracket is defined by
        \[[(x,a),(y,b)]=([x,y]_{\mf{so}(8)},x\cdot b),\]
        where $\cdot$ is the action of $\mf{so}(8)$ on $\mathbf{28}$. The $Y$-tensor comes from the action of $E_{7(7)}\times\mb R^+$ on $\mf a$, which corresponds to the decomposition $R=\w2\mb R^8\oplus\w2(\mb R^8)^*$ from subsection \ref{subsec:M-exact}. This results in the Leibniz parallelisation of the exceptional tangent bundle over $S^7$, twisted by $F_7=6\on{vol}(S^7)$, corresponding to the consistent truncation found in \cite{DN}. The individual summands $\mf{so}(8)$ and $\mathbf{28}$ produce trivialisations of $TS^7\oplus \w5T^*S^7$ and $\w2T^*S^7\oplus(T^*S^7\otimes\w7T^*S^7)$, respectively.
        
        \section{Alternative approaches, \texorpdfstring{$E_8$}{E8} and exceptional complex geometry}

Let use briefly discuss some possible extensions of and concepts related to the Y-algebroid formulation presented here. 
        
      There exists an alternative (and in general inequivalent) approach to Y-algebroids, which makes the connection to the formulas \eqref{eq:general} more immediate, while moves away slightly from Courant and $G$-algebroids. In this approach, we can define the \emph{Y-algebroids$'$} (prime is used to distinguish them from Y-algebroids) as the tuple consisting of
      \begin{itemize}
        \item a vector bundle $E\to M$
        \item an anchor map $\rho\colon E\to TM$
        \item a tensor $Y\colon E^*\otimes E\to E^*\otimes E$
        \item a map (connection) $\nabla\colon \Gamma(E)\times\Gamma(E)\to\Gamma(E)$ satisfying \eqref{eq:connection} and preserving $Y$,
      \end{itemize}
      such that the bracket \emph{defined} by
      \begin{equation}\label{eq:other}
        [u,v]:=\nabla_u v-\nabla_v u+Y(\nabla u)v
      \end{equation}
      satisfies the Leibniz identity and $L_uY=0$ for all $u\in \Gamma(E)$.
      
      One can drop some of the conditions from this definition, provided suitable further constraints on the form of $Y$ are assumed (c.f.\ the constraints in \cite{BCKT}).
      It is possible to work with this definition and derive most of the previous results of this paper in this context, including the local classification and the algebraic characterisation of consistent truncations.
      
      However, in the present paper we have chosen to focus on Y-algebroids (instead of Y-algebroids$'$)
      since we believe that these provide a more conceptual explanation of the origin of the well-known formulas \eqref{eq:general} or \eqref{eq:other}. Y-algebroids also connect more directly to the previous definitions of Courant algebroids and $G$-algebroids.
      
      It is very natural to ask how the present results might be extended to the case of $E_{8(8)}$, where the relevant representation is $R=\r{248}$. The conceptual issue here is that the ``untwisted bracket'', given by the formula \eqref{eq:general}, ceases to satisfy the Leibniz identity, i.e.\ it no longer defines a Leibniz algebroid. There are known approaches to bypass this problem \cite{HS-E8,CR-E8}, but they seem to inevitably lead beyond the realm of Leibniz algebroids. This breakdown is not totally unexpected, since in the case of $E_8$ a proper treatment of the dual graviton becomes unavoidable, which in turn suggests that the required constructions might be radically different from the standard ones.

      Finally we note how the $E_7$ exceptional complex geometry \cite{ASCTW} fits in the present framework. First, decompose $\hat Y\colon E\otimes E\to E\otimes E$ into a symmetric and antisymmetric part, $\hat Y=\hat S+\hat A$, and recall that $\on{rank}(\on{im}\hat A)=1$.
      
      An \emph{exceptional Dirac structure} on an $E_7$-algebroid is a subbundle $F\subset E$ which is
      \begin{itemize}
        \item maximally isotropic, i.e.\ $\hat S|_F=0$ and $\on{rank}F=7$
        \item involutive, i.e.\ $[\Gamma(F),\Gamma(F)]\subset\Gamma(F)$.
      \end{itemize}
      A \emph{complex exceptional Dirac structure} is a subbundle $F\subset E\otimes \mb C$ satisfying the same conditions (with comple\-xified bracket and $Y$).
      An \emph{exceptional complex structure} is a complex exceptional Dirac structure, which in addition satisfies $F\cap \bar F=0$ and $h(u,v):=i\hat A(u,\bar v)$ descends to a definite hermitian inner product on $F$ (valued in the line bundle $\on{im}(\hat A)\otimes\mb C$).
      
      Introduced and analysed in \cite{ASCTW,SW}, these structures provide a convenient framework for investigating various features of general Minkowski $N=1$ $D=4$ flux compactifications in string theory.

    \section{Conclusions and outlook}
    In this paper, we introduced the framework of Y-algebroids, a particular subclass of Leibniz algebroids, that capture the symmetry structures of warped compactifications of M-theory and type II string theory down to $D \geq 4$ dimensions, as manifested in the $O(d,d)$- and $E_{n(n)} \times \mb R^+$-generalised geometries for $n \leq 7$. Moreover, Y-algebroids lead to a compact description of consistent truncations to maximal gauged supergravities in $D \geq 4$ dimensions, of Poisson--Lie U-duality \cite{Sakatani,MT,MST}, and of exceptional complex structures \cite{ASCTW}.
    
    Our concept of Y-algebroids is inspired by the previous notion of $G$-algebroids \cite{BHVW}, which was constructed around a specific group $G$ and included the Leibniz algebroids arising in $E_{n(n)} \times \mb R^+$-generalised geometries for $n \leq 6$, relevant to warped compactifications to $D \geq 5$ dimensions. $G$-algebroids themselves arose as a generalisation of Courant algebroids, replacing the Courant algebroid inner product by a symmetric bilinear map valued in (a vector bundle associated to) a representation of the group $G$. This close resemblance to the Courant case enabled a relatively straightforward analysis and simple proofs of various facts related to dualities and consistent truncations. However, the structure proved to not be general enough to capture the $E_{7(7)} \times \mb R^+$ exceptional geometry.
    
    In contrast, Y-algebroids directly use, as one of their defining features, the tensor \[Y\colon E^* \otimes E \rightarrow E^* \otimes E,\] employed in the description of symmetries of exceptional geometry \cite{BCKT}. Although it makes more immediate the connection to the known explicit formulas for the bracket on $\mb TM$, it leads to a rather radical change of perspective from the viewpoint of Courant or $G$-algebroid geometry and, in particular, increases the computational difficulty of the subsequent analysis. Still, as we demonstrate in the present work, it can be conveniently used to describe the $E_{7(7)} \times \mb R^+$ geometry and to generalise the structural results known from the $n<7$ setups. This is made possible by the fact that no symmetry of $Y$ is assumed, thus allowing it to include the antisymmetric component relevant to the $E_{7(7)}$ group, as well as the setups \eqref{eq:non_split}.
    In addition, Y-algebroids recover naturally the cases $n<7$ as well as various other instances of $G$-algebroids, with $G$ corresponding to the automorphism group of the tensor $Y$.
    
    We have also shown that M-exact $E_{7}$-algebroids can be classified using a set of fluxes $F_1$, $F_4$, $F_7$ satisfying Bianchi identities that allow us to identify them with M-theory fluxes on the 7-dimensional compactification. In particular this leads to a global classification in terms of \v Cech cohomology. This is similar to the classification of Courant algebroids in terms of $H^3(M)$, although the M-exact $E_7$ case is more subtle due to multiple fluxes appearing with nested Bianchi identities.
    
    Several interesting questions follow from this work. Perhaps the most obvious is how to extend this construction to the $E_8$ case, which no longer fits in the framework of Leibniz algebroids \cite{HS-E8,CR-E8}, or indeed, to try and push to the infinite-dimensional cases such as $E_9$ \cite{BCKPS}. Going beyond the $E_{n}$ and $O(p,n-p)$ series, it would be interesting to understand the space of Y-algebroids more generally. In the string and M-theory cases, the relevant Y-algebroids seem to always be of class $(R,y)$, such that the $Y$-tensor ``looks the same'' everywhere on $M$. What is the physics of Y-algebroids that go beyond this?  An interesting observation here is that in the string and M-theory cases, the condition \eqref{eq:implies} always holds. This is an important property, used in the algebraic characterisation of consistent truncations, suggesting that in general one may wish to restrict to only those Y-algebroids satisfying \eqref{eq:implies}.

    There are many other questions. What is the global structure of the groups $\mc G(\mc U)$ that arise in the global classification? How do we calculate $\check H(M,\mc G)$, e.g.\ can we express it in terms of ordinary cohomology classes on the base manifold $M$, as in the analogous Courant case \cite{let,Severa2}? We described how our framework captures Poisson--Lie U-duality and exceptional complex structures. Can this be used to find new Poisson--Lie U-dual pairs, of which currently only few examples are known \cite{BTZ,BZ,BZ2}? Similarly, can our framework be used to gain insights into the global structure of exceptional complex structures? We leave these questions for future work.

\subsection*{Acknowledgements}
    The authors would like to thank Anthony Ashmore for a helpful discussion and the organisers of the ``Supergravity, Generalized Geometry and Ricci Flow'' programme at the Simons Center for Geometry and Physics for hospitality while this work was being completed. The authors also thank the anonymous referee for their valuable comments and suggestions which helped to improve the work. E.\,M.\ also thanks Imperial College London for hospitality. O.\,H.\ was supported by the FWO-Vlaanderen through the project G006119N and by the Vrije Universiteit Brussel through the Strategic Research Program ``High-Energy Physics''. E.\,M.\ was supported by the Deutsche Forschungsgemeinschaft (DFG, German Research Foundation) via the Emmy Noether program ``Exploring the landscape of string theory flux vacua using exceptional field theory'' (project number 426510644). F.\,V.\ was supported by the Postdoc Mobility grant P500PT\underline{\phantom{k}}203123 of the Swiss National Science Foundation. D.\,W.\ was supported in part by the STFC Consolidated Grant ST/T000791/1 and the EPSRC New Horizons Grant EP/V049089/1.


\begin{thebibliography}{99}
  \bibitem{Hitchin} N. Hitchin, \emph{Generalized Calabi-Yau manifolds}, Quart. J. Math. 54 (2003) 281--308.
  \bibitem{Gualtieri} M. Gualtieri, \emph{Generalized Complex Geometry}, Oxford University DPhil thesis (2004).
  \bibitem{CSCW0} A. Coimbra, C. Strickland-Constable, D. Waldram, \emph{Supergravity as Generalised Geometry I: Type II Theories}, JHEP 11 (2011) 091.
  \bibitem{Siegel} W.~Siegel, \emph{Superspace Duality in Low-Energy Superstrings}, Phys. Rev. D 48 (1993), 2826.
  \bibitem{HHZ} O.~Hohm, C.~Hull and B.~Zwiebach, \emph{Generalized Metric Formulation of Double Field Theory}, JHEP 08 (2010), 008.
    \bibitem{JLP} I.~Jeon, K.~Lee and J.~H.~Park, \emph{Differential Geometry with a Projection: Application to Double Field Theory}, JHEP 04 (2011) 014.
  \bibitem{Hull} C. M. Hull, \emph{Generalised Geometry for M-Theory}, JHEP 07 (2007) 079.
  \bibitem{PW} P. P. Pacheco, D. Waldram, \emph{M-theory, exceptional generalised geometry and superpotentials}, JHEP 09 (2008) 123.
  \bibitem{CSCW} A. Coimbra, C. Strickland-Constable, D. Waldram, \emph{$E_{d(d)}\times\mathbb{R}^+$ generalised geometry, connections and M theory}, JHEP 02 (2014) 054.
  \bibitem{LWX} Z. J. Liu, A. Weinstein, P. Xu, \emph{Manin triples for Lie bialgebroids}, J. Diff. Geom. 45 (1997) 3, 547--574.
  \bibitem{KS} C. Klim\v cík, P. Ševera, \emph{Dual non-Abelian T-duality and the Drinfeld double}, Phys. Lett. B 351 (1995), 455--462.
  \bibitem{let} P. \v Severa, \emph{Letters to Alan Weinstein about Courant algebroids}, 1998--2000, arXiv:1707.00265.
  \bibitem{Scs} P. \v Severa, \emph{Poisson-Lie T-duality as a boundary phenomenon of Chern-Simons theory}, JHEP 05 (2016) 044.
  \bibitem{BHVW} M. Bugden, O. Hul\'ik, F. Valach, D. Waldram, \emph{G-algebroids: a unified framework for exceptional and generalised geometry, and Poisson-Lie duality}, Fortsch. Phys. 69 (2021) 4--5, 2100028.
  \bibitem{BHVW2} M. Bugden, O. Hul\'ik, F. Valach, D. Waldram, \emph{Exceptional Algebroids and Type IIB Superstrings}, Fortschr. Phys. 2021, 2100104.
  \bibitem{HV} O. Hul\'ik, F. Valach, \emph{Exceptional Algebroids and Type IIA Superstrings}, Fortsch. Phys. 70 (2022) 6, 2200027.
  \bibitem{I} G. Inverso, \emph{Generalised Scherk-Schwarz reductions from gauged supergravity}, JHEP 12 (2017) 124.
  \bibitem{Sakatani} Y. Sakatani, \emph{$U$-duality extension of Drinfel’d double}, PTEP 2020 (2020) 2, 023B08.
  \bibitem{MT} E. Malek, D. C. Thompson, \emph{Poisson-Lie U-duality in exceptional field theory}, JHEP 04 (2020) 058.
  \bibitem{MST} E.~Malek, Y.~Sakatani, D.~C.~Thompson, \emph{E$_{6(6)}$ exceptional Drinfel'd algebras}, JHEP 01 (2021) 020.
  \bibitem{Pradines} J. Pradines, \emph{Théorie de Lie pour les groupoïdes différentiables. Calcul différentiel dans la catégorie des groupoïdes infinitésimaux}, C. R. Acad. Sci. Paris, Sér. A--B, 264, 1967, A245--A248.
  \bibitem{HS-E8} O. Hohm, H. Samtleben, \emph{Exceptional field theory. III. E$_{8(8)}$}, Phys. Rev. D 90 (2014), 066002.
  \bibitem{CR-E8} M. Cederwall, J. A. Rosabal, \emph{E$_{8}$ geometry}, JHEP 07 (2015) 007.
  \bibitem{BCKT} D. S. Berman, M. Cederwall, A. Kleinschmidt, D. C. Thompson, \emph{The gauge structure of generalised diffeomorphisms}, JHEP 01 (2013) 064.
  \bibitem{DD} T. Dereli, K. Do\u gan, \emph{`Anti-commutable' local pre-Leibniz algebroids and admissible connections}, J. Geom. Phys. 186 (2023) 104752.
  \bibitem{GKP} J. Grabowski, D. Khudaverdyan, N. Poncin, \emph{The supergeometry of Loday algebroids}, The Journal of Geometric Mechanics 5, 185-21 (2011).
  \bibitem{JV} B. Jur\v co, J. Vysok\'y, \emph{Leibniz algebroids, generalized Bismut connections and Einstein–Hilbert actions},     J. Geom. Phys. 97 (2015) 25--33.
  \bibitem{CLS} Z. Chen, Z. Liu, Y. Sheng, \emph{E-Courant algebroids}, Int. Math. Res. Not. 2010 (2010) 4334--4376.
  \bibitem{B} D. Baraglia, \emph{Conformal Courant Algebroids and Orientifold T-Duality}, Int. J. Geom. Meth. Mod. Phys. 10 (2013) 1250084.
  \bibitem{GS} M. Gr\"utzmann, T. Strobl, \emph{General Yang–Mills type gauge theories for $p$-form gauge fields: From physics-based ideas to a mathematical framework or From Bianchi identities to twisted Courant algebroids}, Int. J. Geom. Meth. Mod. Phys. 12 (2014) 1550009.
  \bibitem{BCKPS} G.\ Bossard, M.\ Cederwall, A.\ Kleinschmidt, J.\ Palmkvist, H.\ Samtleben, \emph{Generalized diffeomorphisms for $E_9$},     Phys.\ Rev.\ D 96 (2017) 10, 106022.
  \bibitem{CP} M.\ Cederwall, J.\ Palmkvist, \emph{Extended geometries}, JHEP 02 (2018) 071.
  \bibitem{PS} J.-H.\ Park, Y.\ Suh, \emph{U-gravity: SL(N)}, JHEP 06 (2014) 102.
  \bibitem{LSCW} K. Lee, C. Strickland-Constable, D. Waldram, \emph{Spheres, generalised parallelisability and consistent truncations}, Fortsch. Phys. 65 (2017) 10--11, 1700048.
  \bibitem{HS} O. Hohm, H. Samtleben, \emph{Consistent Kaluza-Klein Truncations via Exceptional Field Theory}, JHEP 01 (2015) 131.
  \bibitem{EGM} C. Eloy, M. Galli, E. Malek, \emph{Adding fluxes to consistent truncations: IIB supergravity on $AdS_3\times S^3\times S^3\times S^1$}, arXiv:2306.12487.
  \bibitem{SC} C. Strickland-Constable, \emph{Subsectors, Dynkin Diagrams and New Generalised Geometries}, JHEP 08 (2017) 144.
  \bibitem{BCKPSS} G.\ Bossard, M.\ Cederwall, A.\ Kleinschmidt, J.\ Palmkvist, E.\ Sezgin, L.\ Sundberg, \emph{Extended geometry of magical supergravities}, JHEP 05 (2023) 162.
  \bibitem{GST} M.\ Gunaydin, G.\ Sierra, P.\ K.\ Townsend, \emph{Exceptional Supergravity Theories and the MAGIC Square}, Phys.\ Lett.\ B 133 (1983) 72--76.
  \bibitem{GST2} M.\ Gunaydin, G.\ Sierra, P.\ K.\ Townsend, \emph{The Geometry of N=2 Maxwell-Einstein Supergravity and Jordan Algebras}, Nucl.\ Phys.\ B 242 (1984) 244--268.
  \bibitem{GSS} M.\ Gunaydin, H.\ Samtleben, E.\ Sezgin, \emph{On the Magical Supergravities in Six Dimensions}, Nucl.\ Phys.\ B 848 (2011) 62--89.
  \bibitem{HasSak} F.~Hassler, Y.~Sakatani, \emph{All maximal gauged supergravities with uplift}, arXiv:2212.14886.
  \bibitem{DN} B. de Wit, H. Nicolai, \emph{The Consistency of the $S^7$ Truncation in $D=11$ Supergravity}, Nucl. Phys. B 281 (1987) 211--240.
  \bibitem{ASCTW} A. Ashmore, C. Strickland-Constable, D. Tennyson, D. Waldram, \emph{Generalising $G_2$ geometry: involutivity, moment maps and moduli}, JHEP 01 (2021) 158.
    \bibitem{SW} G.~R.~Smith and D.~Waldram, \emph{M-theory Moduli from Exceptional Complex Structures}, arXiv:2211.09517.
    \bibitem{Severa2} P. Ševera, \emph{Poisson-Lie T-Duality and Courant Algebroids}, Lett. Math. Phys. 105 (2015) 12, 1689--1701.
  \bibitem{BTZ} C. D. A. Blair, D. C. Thompson, S. Zhidkova, \emph{Exploring Exceptional Drinfeld Geometries}, JHEP 09 (2020) 151.
  \bibitem{BZ} C. D. A. Blair, S. Zhidkova, \emph{Generalised U-dual solutions in supergravity}, JHEP 05 (2022) 081.
  \bibitem{BZ2} C. D. A. Blair, S. Zhidkova, \emph{Generalised U-dual solutions via ISO(7) gauged supergravity}, JHEP 12 (2022) 093.
\end{thebibliography}
\end{document}